\documentclass[aip, reprint, cha]{revtex4-1}

\usepackage{amsmath}
\usepackage{amssymb}
\usepackage{braket}
\usepackage{bbm}
\usepackage{graphicx}
\usepackage{tikz}
\usepackage{caption}
\usepackage{subcaption}
\usepackage{hyperref}
\usepackage{verbatim}
\hypersetup{
    colorlinks=true,
    linkcolor=blue,
    filecolor=blue,      
    urlcolor=blue,
    citecolor = blue,
}

\allowdisplaybreaks

\usepackage{array}
\usepackage{multirow}
\usepackage{booktabs}
\usepackage{comment}
\usepackage{natbib}
\usepackage{amsthm}
\newtheorem{theorem}{Proposition}

\newtheorem{lemma}{Lemma}[theorem]
\newtheorem{definition}{Definition}
\usepackage[normalem]{ulem}
\usepackage{bm}
\usepackage[utf8]{inputenc}
\usepackage[T1]{fontenc}
\usepackage{mathptmx}
\usepackage{etoolbox}

\makeatletter
\def\@email#1#2{%
 \endgroup
 \patchcmd{\titleblock@produce}
  {\frontmatter@RRAPformat}
  {\frontmatter@RRAPformat{\produce@RRAP{*#1\href{mailto:#2}{#2}}}\frontmatter@RRAPformat}
  {}{}
}%
\makeatother
\begin{document}
\preprint{AIP/123-QED}

\title[Diffusion as a Signature of Chaos]{Diffusion as a Signature of Chaos}
\author{Nachiket Karve}
\email{nachiket@bu.edu}
\author{Nathan Rose}
\author{David Campbell}
\affiliation{Department of Physics, Boston University, Boston, Massachusetts 02215, USA}

\date{\today}

\begin{abstract}
    While classical chaos is defined via a system's sensitive dependence on its initial conditions (SDIC), this notion does not directly extend to quantum systems. Instead, recent works have established defining both quantum and classical chaos via the sensitivity to adiabatic deformations and measuring this sensitivity using the adiabatic gauge potential (AGP). Building on this formalism, we introduce the ``observable drift" as a probe of chaos in generic, non-Hamiltonian, classical systems. We show that this probe correctly characterizes classical systems that exhibit late-time SDIC as chaotic. Moreover, this characterization is consistent with the measure-theoretic definition of chaos via weak mixing. Thus, we show that these two notions of sensitivity (to changes in initial conditions and to adiabatic deformations) can be probed using the same quantity, and therefore, are equivalent definitions of chaos. Numerical examples are provided via the tent map, the logistic map, and the Chirikov standard map.
\end{abstract}

\maketitle

\begin{quotation}
In 1972, the mathematician Ed Lorenz coined the popular term ``butterfly effect” to describe how tiny perturbations in a system’s initial conditions, such as the fluttering of a butterfly’s wings, could ultimately lead to large-scale outcomes like a tornado \cite{lorenz_1972}. This property of a system being sensitive to tiny perturbations has traditionally been used to define classical chaos. However, this notion of chaos does not manifest directly in quantum systems, since they do not exhibit such sensitivity. Instead, a quantum system is considered chaotic when its classical counterpart is chaotic \cite{berry_1989}. In an attempt to unify our understanding of both quantum and classical chaos, recent works have proposed an alternative definition of chaos via the sensitivity to adiabatic deformations \cite{pandey_2020,lim_2024}. According to this view, a system (either quantum or classical) is chaotic when a tiny change to its Hamiltonian brings about large changes in its stationary states. In this work, we show that this new definition of chaos is essentially equivalent to our traditional understanding of chaos. To demonstrate this, we introduce the \textit{observable drift} as a probe of chaos in generic dynamical systems, which is a generalization of the probe used to measure sensitivity to adiabatic deformations in Hamiltonian systems. We then demonstrate that this probe accurately classifies classical systems exhibiting irreversible dynamics as chaotic. The underlying reason for this equivalence is that both sensitivities – to initial conditions and to adiabatic deformations – manifest as decaying correlations between observables, thereby linking two seemingly distinct characterizations of chaos under a common framework.
\end{quotation}

\section{Introduction}

\begin{figure}
    \centering
    \includegraphics[]{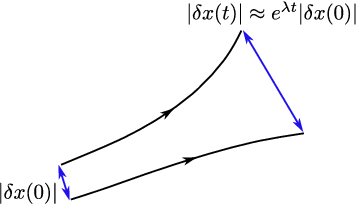}
    \caption{Two nearby trajectories of a chaotic system separating away from each other exponentially fast.}
    \label{fig_lyapunov}
\end{figure}

Sensitive dependence on the initial conditions (SDIC) \cite{guckenheimer_1979,glasner_1993} is a defining feature of chaos in classical dynamical systems. A system exhibits SDIC if a small change in its initial conditions leads to a dramatically different outcome in the future. One of the most common ways to measure this sensitivity is via the maximal Lyapunov exponent $\lambda$ \cite{mawhin_2005,oseledets_1968}. This exponent is defined as
\begin{equation}
    \lambda = \lim_{t\to\infty} \ln \left|\frac{\delta x(t)}{\delta x(0)}\right|,
\end{equation}
where $\delta x$ denotes the separation between two nearby trajectories of the system in its phase-space (see Fig. \ref{fig_lyapunov}). A positive value indicates exponentially fast separation of nearby trajectories, and such systems are considered to be \textit{chaotic in the Lyapunov sense}. Additionally, chaotic systems become unpredictable beyond the Lyapunov time, which is inversely proportional to the Lyapunov exponent, $t_{\mathrm{Lyapunov}} = 1/\lambda$ \cite{vallejo_2017}. Systems with multi-dimensional phase spaces are characterized by a spectrum of Lyapunov exponents, where the largest exponent is usually used as an indicator of chaos time-scales \cite{benettin_1980}.

It is important to note that the Lyapunov exponent only captures the growth rate of finite perturbations at short times. Consider two states of a chaotic system that are close to each other initially. As both states evolve with time, the separation between them grows as $|\delta x(t)| \approx e^{\lambda t}|\delta x(0)|$. However, this expression is valid only when $|\delta x(t)|$ is small, and as the separation grows large at later times, nonlinear effects begin to take effect. Thus, the Lyapunov exponent can sometimes fail to capture late-time dynamics. An example of this failure is demonstrated in Fig. \ref{fig_tent}, where two trajectories of the symmetric tent map \cite{crampin_1994} (defined and discussed in more detail in later sections) are studied. The two states initially start close to each other, and the Lyapunov exponent correctly describes how the separation grows over short times. However, this growth does not continue indefinitely, and the separation alternates between periods of growth and decay. In fact, we observe that the two states remain close to each other even at times much longer than the Lyapunov time. Unsurprisingly, persistent periodic correlations are observed in this regime of the tent map \cite{shigematsu_1983,yoshida_1983,dorfle_1985}. Another system in which the Lyapunov exponent fails to capture non-linear effects is the system studied by Perron \cite{perron_1930,leonov_2003}. Perron demonstrates that it is possible to construct systems with positive Lyapunov exponents, where nearby trajectories do not diverge from each other, and systems with negative Lyapunov exponents, where nearby trajectories do diverge. 

Moreover, the Lyapunov exponent does not indicate how chaos is manifested in different observables of the system. In many-body systems, different observables can show different late-time and transient behavior. Consider the inner solar system, for example, where the Lyapunov time is estimated to be about $5$ Myrs \cite{mogavero_2023}. However, instabilities at this time-scale are only observed in the angle variables of the system, while the shapes and relative positions of the planetary orbits stay stable for far longer \cite{ito_2002,hayes_2007}. In fact, the Lyapunov time is much shorter than the estimated age of the solar system ($\sim 5$ Gyrs \cite{bouvier_2010}), as well as the dynamical half-life of Mercury ($\sim30-40$ Gyrs \cite{mogavero_2021}). Thus, different observables can become unstable at different times, and therefore, defining a ``chaos time-scale'' is inherently observable dependent. Such behavior has also been observed in the Fermi-Pasta-Ulam-Tsingou (FPUT) system \cite{fermi_1955}, where the energy variables exhibit quasi-periodic behavior far beyond the Lyapunov time \cite{benettin_2008,benettin_2018,kevin_2023,lando_2025,karve_2025}. This system is characterized by the formation of a non-thermalizing, metastable state \cite{zabusky_1965, flach_2005, zabusky_2005, flach_2006, pace_2019, karve_2024, deng_2025}. The KAM theorem \cite{kolmogorov_1979, arnold_1963} asserts that this behavior is fairly generic: that is, for systems with sufficiently weak integrability-breaking, the action variables can appear non-thermalizing almost indefinitely for a large fraction of initial conditions.

\begin{figure}
    \centering
    \includegraphics[width=0.95\linewidth]{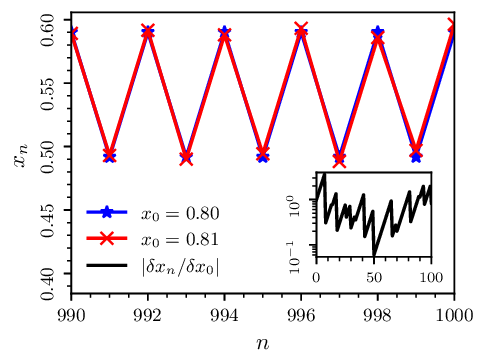}
    \caption{Two trajectories in a chaotic tent map (non-linearity $r=1.2$) with nearby initial conditions $x_0 = 0.8$, and $x_0 = 0.81$. Although the Lyapunov time in this regime is $t_{\text{Lyapunov}} \sim 5.48$, the trajectories remain close indefinitely. Inset: The separation between the two trajectories, $|\delta x_n/\delta x_0|$, as a function of time, $n$. The separation alternates between periods of exponential growth and decay. The growth rate is specified by the Lyapunov exponent.}
    \label{fig_tent}
\end{figure}

Consequently, we will adopt a statistical view of classical chaos through the concept of \textit{weak mixing} \cite{walters_1982,sinai_1982,viana_2016}. Instead of defining chaos through the behavior of individual trajectories, we define chaos through the behavior of ensembles of the system. Generic localized ensembles of a mixing system eventually thermalize and spread over the accessible phase-space. This implies that any two states, no matter how close they are to each other initially, are eventually separated from each other in the thermal ensemble. Therefore, trajectories in mixing systems exhibit SDIC at late times \cite{he_2004,werndl_2009}. Correlations between observables are known to decay in mixing systems; thus, how this chaos manifests in a particular observable can be quantified through the rate of decay of its autocorrelation function. We will refer to mixing systems as being \textit{chaotic in the mixing sense}.

The trajectory-based view of classical chaos does not translate well to quantum systems. There is no direct analogue of SDIC in quantum mechanics due to a lack of phase-space trajectories. Therefore, a quantum system is considered to be chaotic only when its semiclassical limit is chaotic and exhibits SDIC \cite{berry_1989,jensen_1992,maldacena_2016,andreev_1996}. Traditional probes of quantum chaos have focused on the spectral properties of the system, such as the level spacing statistics \cite{wigner_1957,bgs_1984,berry_1977,srednicki_1999,d'alessio_2016}. However, more recent studies have established that quantum chaos can be defined via the system's sensitivity to adiabatic changes \cite{pandey_2020, leblond_2021, pozsgay_2024, kolodrubetz_2017, vidmar_2025, orlov_2023, sharipov_2024, peres_1984}. Furthermore, this sensitivity is probed through the generator of adiabatic deformations, called the adiabatic gauge potential (AGP). Note that this sensitivity is necessarily observable dependent, since it depends on which parameter of the system is being changed. This formalism is also well-defined in classical Hamiltonian systems \cite{lim_2024,bermudez_2024}, and in \cite{karve_2025} it was shown that this probe is related to the diffusion of time-integrated perturbations. We will refer to such quantum and classical systems which are sensitive to adiabatic deformations as being \textit{chaotic in the adiabatic sense}.

Nevertheless, in this work, we demonstrate how these two notions of chaos -- in the mixing sense and in the adiabatic sense -- can be reconciled. We introduce the ``observable drift'', a generalization of the formalism developed in \cite{karve_2025}, which allows us to probe SDIC in classical systems as well as the sensitivity to adiabatic deformations in Hamiltonian systems. We also show that this probe allows us to classify both transient and asymptotic dynamics of the system into four distinct classes, namely, dissipative, regular, strongly chaotic, and weakly chaotic. Regular and dissipative dynamics are characterized by the absence of diffusion in the time-integrated observable, whereas strong and weak chaos correspond to normal and anomalous diffusion, respectively.  We show that absolutely decaying correlations, a feature of weak mixing, cause this diffusion and that the rate of decay determines the strength of chaos. Simple examples of discrete-time maps are provided to demonstrate the efficacy of this method.

The remainder of this paper is organized as follows. Sec. \ref{sec_weakMixing} reviews some essential concepts from ergodic theory and discusses mixing systems. Sec. \ref{sec_agp} discusses defining chaos via the sensitivity to adiabatic deformations and how it is probed via the AGP. These two perspectives on chaos are united under a common framework of the observable drift, as discussed in Sec. \ref{sec_drift}. Sec. \ref{sec_numerics} demonstrates how this method can be applied to some very simple systems. We summarize our results and conclude in Sec. \ref{sec_conclusions}. Appendices \ref{app_ergodicTheory}, \ref{app_attractors}, \ref{app_ergodic}, \ref{app_01Test}, and \ref{app_numerics} discuss formal mathematical arguments.

\section{Chaos in the Mixing Sense}
\label{sec_weakMixing}

\begin{figure*}
    \centering
    \begin{subfigure}[t]{0.65\linewidth}
        \centering
        \includegraphics[]{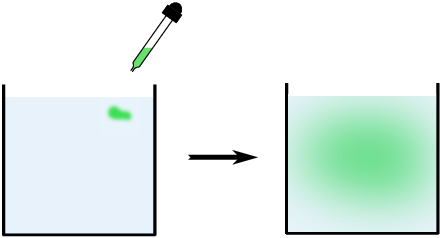}
        \caption{}
        \label{fig_mixingInk}
    \end{subfigure}
    \begin{subfigure}[t]{0.31\linewidth}
        \centering
        \includegraphics[]{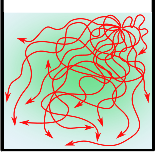}
        \caption{}
        \label{fig_mixingInkTraj}
    \end{subfigure}
    \begin{subfigure}[t]{0.65\linewidth}
        \centering
        \includegraphics[]{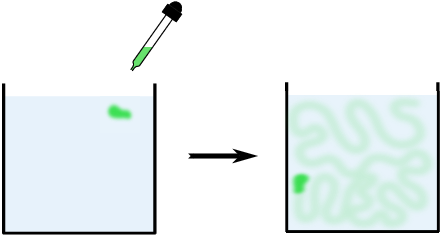}
        \caption{}
        \label{fig_ergodicButNotMixing}
    \end{subfigure}
    \begin{subfigure}[t]{0.31\linewidth}
        \centering
        \includegraphics[]{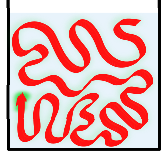}
        \caption{}
        \label{fig_ergodicButNotMixingTraj}
    \end{subfigure}
    \caption{(a) Evolution of an ink drop in water. While the drop is initially a highly localized object, it eventually spreads out over the entire liquid. (b) Trajectories of different particles in the drop, illustrated in red. Initially, all of them start close to each other, but eventually they spread over large distances. (c) A hypothetical ink drop that is ergodic but not mixing. The drop stays localized indefinitely, although it explores the entire liquid. (d) Trajectories of all ink particles in this drop always remain close to each other.}
\end{figure*}

The statistical viewpoint of chaos shifts our focus from the divergence of individual trajectories to the evolution of probability distributions over the phase space. Instead of asking how two nearby initial points separate with time, we ask how an ensemble of initial conditions spreads and mixes. In this section, we introduce concepts from ergodic theory that underlie this statistical description, namely weak and strong mixing, and discuss how they characterize long-time chaos. A more detailed and formal mathematical treatment of these concepts is provided in Appendix \ref{app_ergodicTheory}. Although we will primarily focus on discrete-time systems, that is, systems whose evolution is described by a map $x_{n+1} = f(x_n)$, all arguments can easily be extended to continuous-time systems.

Starting with an ensemble of states of the system, we represent this ensemble after $n$ iterations of the map by the probability distribution $\rho_n(x)$. The evolution of $\rho_n$ is governed by the Frobenius-Perron operator, given by
\begin{equation}
    \rho_{n+1}(x) = \mathcal{L}\rho_n(x) = \int dy \ \rho_n(y) \delta(x - f(y)).
\end{equation}
Oftentimes, the time-average of $\rho_n$ converges to a stationary distribution, which we call the ``natural'' distribution. It is given by
\begin{equation}
    \bar{\rho}(x) = \lim_{N\to\infty} \frac{1}{N+1} \sum_{n=0}^N \rho_n(x).
\end{equation}
One can easily check that the natural distribution is stationary, since $\bar{\rho} = \mathcal{L}\bar{\rho}$. Properties such as ergodicity and mixing are defined with respect to a particular natural distribution, and as such, the same system can exhibit different behaviors when different distributions are considered.

A system is considered to be ergodic with respect to a certain natural distribution if a generic state explores the entire phase-space, where the probability of visiting any region is weighted by the distribution. This results in time-averages of typical observables being equal to ensemble averages. That is, for an observable $O$,
\begin{equation}
    \lim_{N\to\infty} \frac{1}{N+1} \sum_{n=0}^N O(x_n) = \int dx \ \bar{\rho}(x) O(x).
\end{equation}

Mixing, on the other hand, is a stronger property than ergodicity. While ergodicity ensures time-averages and ensemble averages coincide, it is still possible for the system to remember its initial state indefinitely. Mixing, however, is an irreversible process that ensures a loss of memory. One can understand mixing by drawing a simple analogy with a drop of ink spreading in water, as shown in Fig. \ref{fig_mixingInk}. Here, the phase space is represented by the water, and the ink particles represent states in the ensemble. In such a mixing system, we start with a highly localized ensemble that eventually spreads out over the entire phase space. Information about the initial state is completely lost since it is impossible to reverse engineer the initial ink drop. A crucial property of \textit{strongly mixing systems} is the loss of correlations between observables. That is,
\begin{equation}
    \lim_{n\to\infty} C_n(O_1, O_2) = 0,
    \label{eq_strongMixing}
\end{equation}
where the correlation function $C_n(O_1, O_2)$ between observables $O_1$ and $O_2$ is defined as
\begin{align}
    C_n(O_1, O_2) =& \int dx \ \bar{\rho}(x) O_1(f^n(x))O_2(x)\notag\\& - \left(\int dx \ \bar{\rho}(x)O_1(x)\right)\left(\int dx \ \bar{\rho}(x)O_2(x)\right).
\end{align}
The condition in Eq. \ref{eq_strongMixing} is slightly relaxed in \textit{weakly mixing systems}, where
\begin{equation}
    \lim_{N\to\infty} \frac{1}{N+1}\sum_{n=0}^N|C_n(O_1, O_2)| = 0.
\end{equation}
Note that strong mixing is a subset of weak mixing. We collectively refer to systems possessing either property as mixing systems.

Importantly, mixing systems exhibit SDIC at long times. This can be seen by considering the ink drop analogy again and tracing out trajectories of different ink particles, as shown in Fig. \ref{fig_mixingInkTraj}. Every trajectory starts close together in the drop, but spreads away from each other asymptotically. If the drop represents a small error in prescribing the initial condition of the system, then mixing implies that the error is eventually magnified to be as large as the entire system (see Appendix \ref{app_ergodicTheory} for a detailed discussion).

To better understand the difference between ergodicity and mixing, one can consider a hypothetical ink drop that is ergodic but not mixing, as shown in Fig. \ref{fig_ergodicButNotMixing}. Such a drop would stay localized indefinitely, while exploring the entire container. Unlike the mixing system, it does not spread throughout the liquid, thus remembering its past. If the trajectories of the ink particles in this drop are traced out (Fig. \ref{fig_ergodicButNotMixingTraj}), one sees that they always remain close to each other. Thus, such a system does not exhibit SDIC.

Thus, chaos in the mixing sense captures long-time statistical irreversibility of deterministic motion. It describes how ensembles of the system spread throughout the phase space, and how this leads to a loss of correlations between observables. As we will see in Sec. \ref{sec_drift}, the rate of decay of these correlations is related to the ``strength" of the chaos in the system. 

\section{Chaos in the Adiabatic Sense}
\label{sec_agp}
While classical chaos is defined in terms of SDIC, no equivalent formulation exists in quantum mechanics. Quantum systems exhibit unitary evolution; therefore, the difference between any two states is always preserved. To quantify sensitivity in such systems, one proposal is to instead examine how the state responds to changes in external parameters. A seminal idea by Pandey, \textit{et al.} established that quantum chaos can be defined and probed via the system's response to slow, adiabatic variations of its parameters. We give a brief description of this idea in this section. See \cite{pandey_2020,kolodrubetz_2017,lim_2024} for a more detailed exposition. 

Let us consider a system that depends on a parameter $\lambda$, and is described by the following Hamiltonian:
\begin{equation}
    H(\lambda) = H_0 + \lambda V.
\end{equation}
The question of the system's sensitivity to adiabatic deformations is equivalent to asking how its eigenstates change under slow changes in $\lambda$. This is captured by the adiabatic gauge potential (AGP) operator $\mathcal{A}_\lambda$, defined as
\begin{equation}
    \mathcal{A}_\lambda\ket{n(\lambda)} = i\hbar \partial_\lambda \ket{n(\lambda)},
\end{equation}
where $\ket{n(\lambda)}$ are the instantaneous eigenstates of the system. A ``large'' AGP norm $||\mathcal{A}_\lambda||^2$ is a sign that the system is sensitive to adiabatic deformations, and therefore, is chaotic in the adiabatic sense.

This formalism can be readily extended to classical systems as shown in \cite{lim_2024,karve_2025}. In classical systems, the AGP becomes a function of the phase-space coordinates under a Wigner-Weyl transformation. One can then track the fluctuations of the AGP along any particular trajectory through
\begin{equation}
    \mathcal{A}_\lambda(t) = \mathcal{A}_\lambda(0) - \int_0^t d\tau \ \left(V(\tau) - \bar{V} \right).
\end{equation}
These fluctuations are divergent when the system is chaotic. In \cite{karve_2025}, it was shown that the fluctuations of the AGP along trajectories in the phase-space are related to the autocorrelation function of the perturbation $V$. Furthermore, a decaying autocorrelation of $V$ was demonstrated to result in diverging AGP fluctuations. Thus, chaos in the mixing sense implies chaos in the adiabatic sense. Note that this probe of chaos is observable dependent: the AGP can exhibit different behaviors for different $V$s.

This connection between decaying correlations and divergent AGP fluctuations provides a quantitative bridge between the statistical and adiabatic descriptions of chaos. However, it is essential to note that the adiabatic interpretation of chaos is only meaningful in Hamiltonian systems, where the notion of adiabatic variation is well-defined. To extend these ideas to more general, non-Hamiltonian dynamics, we introduce in the next section a unified formalism that applies to generic dynamical systems.

\section{Observable Drift}
\label{sec_drift}

The AGP framework provides a powerful measure of sensitivity for Hamiltonian systems, but it cannot be extended to non-Hamiltonian dynamics. To formulate a universal diagnostic of chaos, we instead focus on the behavior of time-integrated observables. Consider some observable $O$ in a discrete-time dynamical system, and suppose that $O$ has a finite long-time average on any trajectory in the phase space:
\begin{equation}
    \bar{O}(x) = \lim_{N\to\infty}\frac{1}{N+1}\sum_{n=0}^N O(f^n(x)).
\end{equation}
This time average is trajectory dependent in a generic system, but is equal to the phase-space average for almost every trajectory in an ergodic system. Taking inspiration from the AGP formalism, we construct a random walk for this observable, by considering the cumulative displacement of $O$ from its mean along a trajectory,
\begin{equation}
    \Delta_N(x) = \sum_{n=0}^N \ \left[O(f^n(x)) - \bar{O}(x)\right].
\end{equation}
We refer to this quantity as the ``observable drift'', or $O$-drift for short.
 
It is easy to see that the evolution of the drift depends on the amount of correlation between different points on the trajectory. In a non-chaotic system, where the motion is predictable, the $O$-drift is typically a regular, smooth function. However, in a highly chaotic system, the values that the observable takes at each point on the trajectory can be completely random and independent of each other. See Fig. \ref{fig_driftSketch} for a comparison between chaotic and non-chaotic observable drifts. This stark contrast enables us to distinguish between chaotic and regular trajectories.

\begin{figure}
    \centering
    \includegraphics[width=0.97\linewidth]{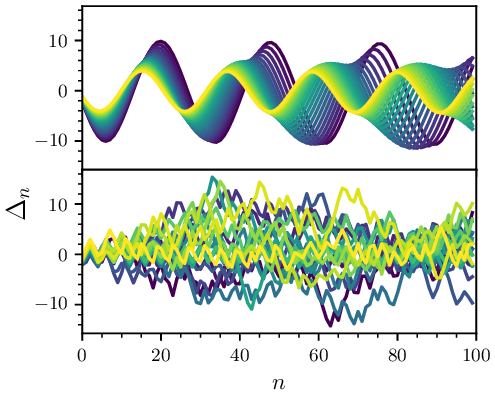}
    \caption{Sketch of the observable drift on an ensemble of non-chaotic (top) and chaotic (bottom) trajectories.}
    \label{fig_driftSketch}
\end{figure}

A central object in the study of random walks is the mean squared displacement, which measures the spread of trajectories in a given ensemble over time \cite{ibe_2013}. In a similar fashion, we study the ensemble-averaged drift variance, $\langle\sigma_N^2\rangle$, where
\begin{equation}
    \sigma_N^2(x) = \frac{1}{N+1}\sum_{n=0}^N \Delta_n^2(x) - \left(\frac{1}{N+1}\sum_{n=0}^N \Delta_n(x)\right)^2,
    \label{eq_driftVar}
\end{equation}
and
\begin{equation}
    \langle \sigma_N^2\rangle = \int dx \ \rho_0(x) \ \sigma_N^2(x)
\end{equation}
denotes an average over a distribution of initial conditions given by the probability density $\rho_0$. This is a direct generalization of the fidelity susceptibility in Hamiltonian systems \cite{lim_2024,karve_2025,delcampo_2012,delcampo_2017}. 

We then define a diffusion exponent $\gamma(N)$ as
\begin{equation}
    \gamma(N) = \frac{\ln \langle \sigma_N^2 \rangle}{\ln N}.
    \label{eq_meanDriftVar}
\end{equation}
This diffusion exponent serves a two-fold purpose. First, its value at finite times captures the behavior of the transient ensemble density $\bar{\rho}_N$, defined as
\begin{equation}
    \bar{\rho}_N(x) = \frac{1}{N+1}\sum_{n=0}^N \rho_n(x).
    \label{eq_avgMeasure}
\end{equation}
Depending on the value of the exponent, this transient dynamics can be categorized into four distinct regimes: dissipative, regular, strongly chaotic, and weakly chaotic. We verify this claim in Sec. \ref{sec_numerics} through numerical examples. These regimes are summarized in Table \ref{tab_class}.

Secondly, as $\bar{\rho}_N$ converges to the natural distribution $\bar{\rho}$, the asymptotic value of the exponent, $\gamma(\infty)$, tells us about the properties of $\bar{\rho}$. Again, this asymptotic value allows us to classify the behavior of the natural distribution as belonging to either one of the regimes listed in Table \ref{tab_class}. 

Thus, the diffusion exponent captures both the transient dynamics of the specific ensemble under study, as well as the dynamics of the equilibrium distribution. In systems with long transient dynamics, a clear switch from transient values of the exponent to equilibrium values can be used to define a crossover time. This time serves as a diagnostic for the equilibration of the ensemble. Importantly, the Lyapunov exponent is agnostic to the initial distribution and fails to capture such transient behavior.

\begin{table}
\begin{ruledtabular}
    \begin{tabular}{lr}
        $\gamma < 0$ & Dissipative \\ 
        $\gamma = 0$ & Regular \\ 
        $\gamma = 1$ & Strong Chaos \\ 
        $1 < \gamma \leq 2$ & Weak Chaos \\ 
    \end{tabular}
\end{ruledtabular}

\caption{The diffusion exponent $\gamma$ for dissipative, regular, strongly chaotic, and weakly chaotic regimes.}
\label{tab_class}    
\end{table}

We now provide an intuitive understanding of the regimes described in Table \ref{tab_class} here. See Appendices \ref{app_attractors} and \ref{app_ergodic} for a rigorous analysis. Let us first consider the dissipative regime, where the diffusion exponent is negative, implying a decaying drift variance. This behavior is observed in systems with attractors (both chaotic and non-chaotic). A sufficiently wide initial distribution in such a system goes through an initial contracting phase as it becomes more localized. In Appendix \ref{app_attractors}, we demonstrate that this initial decay is $\langle\sigma_N^2\rangle \sim \mathcal{O}(1/N)$ in the presence of a fixed point or limit-cycle attractor.

The regular regime is characterized by systems that exhibit periodic or quasi-periodic motion. In such systems, we find the drift variance to saturate at a finite value. For example, although a system with a limit-cycle attractor can initially be dissipative, it exhibits periodic motion asymptotically. Thus, the behavior of the natural distribution in such a system is regular, as shown in Appendix \ref{app_limitCycle}. Similarly, ergodic but not mixing systems exhibit quasi-periodic motion and are therefore also classified as regular systems (see Appendix \ref{app_ergodicNotMixing}). 

On the other hand, decaying correlations in chaotic systems lead to diffusion of the drift (Appendix \ref{app_weakMixing}). If the decay of correlations on the natural measure is sufficiently fast, i.e. $|C_n| < \mathcal{O}(1/n)$, then each step in the random walk becomes almost independent, and the mean drift variance shows linear growth \cite{gyu_2017}. We characterize this normal diffusion of the drift as strong chaos. On the other hand, if the decay of correlations is slow, i.e. $|C_n| > \mathcal{O}(1/n)$, then super-linear growth (anomalous diffusion) of the variance is observed, which we characterize as weak chaos. In an ergodic ensemble, the relation between the decay of the autocorrelation of $O$ and the growth of the drift variance is captured by the following identity:
\begin{equation}
    \langle\sigma^2_N\rangle = \sum_{r = -N}^N \frac{(N-|r|)(N+1-|r|)(N+2-|r|)}{6(N+1)^2} C_r(O,O).
    \label{eq_flucDissText}
\end{equation}
See Appendix \ref{app_ergodic} for a rigorous derivation.

It is important to stress here that weak and strong chaos are not analogous to weak and strong mixing. In fact, any weakly (and consequently, any strongly) mixing system can exhibit either weak or strong chaos. We also note that there are a number of different names used for these regimes in the literature. For example, the weakly chaotic regime is also known as the ``nearly-integrable'', ``weak integrability breaking'', or ``maximally sensitive'' regime \cite{pandey_2020,pozsgay_2024,lim_2024,vanovac_2024,vidmar_2025,surace_2023}.

This framework readily extends to continuous-time systems, where the drift and its variance are defined by replacing the sum over $n$ with an integral over time:
\begin{subequations}
    \begin{equation}
        \Delta(t) = \int_0^t d\tau \ \left[O(\vec{x}(\tau)) - \bar{O}\right],
    \end{equation}
    \begin{equation}
        \sigma^2(t) = \frac{1}{t}\int_0^t d\tau \ \Delta(\tau)^2 - \left(\frac{1}{t}\int_0^t d\tau \ \Delta(\tau)\right)^2.
    \end{equation}
\end{subequations}
As noted in \cite{karve_2025}, this probe is inherently observable dependent. This can prove advantageous, as it allows us to determine the observer-specific behavior.

It should also be noted that this framework is conceptually very similar to the 0-1 test of chaos, introduced by Gottwald, \textit{et al} \cite{gottwald_2004,gottwald_2005,gottwald_2009}. This binary test distinguishes between regular and chaotic dynamics in a time-series $\phi(n)$ by coupling it to a two-dimensional system via
\begin{subequations}
    \begin{equation}
        p(n+1) = p(n) + \phi(n)\cos cn,
    \end{equation}
    \begin{equation}
        q(n+1) = q(n) + \phi(n)\sin cn,
    \end{equation}
\end{subequations}
and analyzing the nature of diffusion in this driven system. Here $c$ is some positive constant in $(0,2\pi)$. In this construction, the mean squared displacement of the vector $(q(n),p(n))$ is calculated, and the authors demonstrate that the corresponding diffusion exponent is $0$ for regular systems and $1$ for chaotic systems. The conceptual similarities between the 0-1 test and our method are obvious, and therefore, we do not expect these two methods to have major disagreements. However, the 0-1 test requires a more complicated construction, since it involves driving an external system. Further, this is a binary test and does not distinguish between strong and weak chaos. It is also not clear whether information about the transient dynamics of the system can be extracted via the 0-1 test. See Appendix \ref{app_01Test} for more details.

\section{Numerical Examples}
\label{sec_numerics}

\begin{figure*}
    \centering
    \begin{subfigure}[t]{0.47\linewidth}
        \centering
        \includegraphics[width=\linewidth]{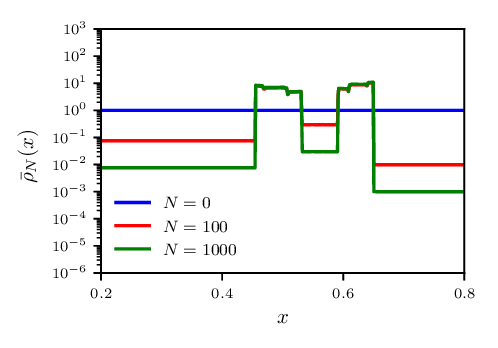}
        \caption{}
    \end{subfigure}
    \begin{subfigure}[t]{0.47\linewidth}
        \centering
        \includegraphics[width=\linewidth]{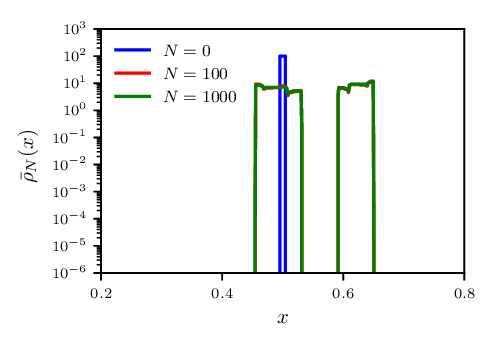}
        \caption{}
    \end{subfigure}
    \caption{Evolution of the time-averaged density of the tent map with $r=1.3$, for the uniform distribution (a), and a narrower initial distribution with width $w=0.01$ and centered at $x_0=0.5$ (b). Convergence to the invariant measure is faster in the latter case, since the initial distribution lies inside the strange attractor.}
    \label{fig_tentDensity}
\end{figure*}

We now provide a few simple examples of discrete maps where our framework can accurately distinguish between dissipative, regular, and chaotic behavior. We demonstrate how the mean drift variance can be used to identify crossovers between different stages of the dynamics. We focus on three different examples here: namely, the tent map \cite{crampin_1994}, the logistic map \cite{may_1976}, and the Chirikov standard map \cite{chirikov_1971,chirikov_1979}. Some finer numerical details are discussed in Appendix \ref{app_numerics}.

\subsection{Tent Map}

\begin{figure}
    \centering
    \includegraphics[width=0.97\linewidth]{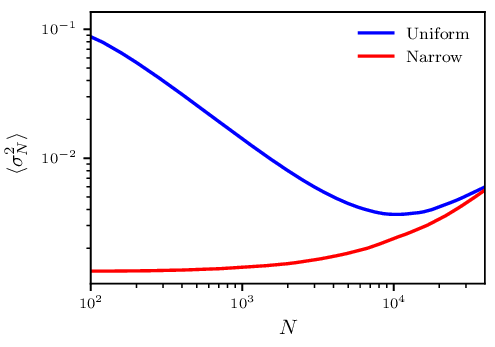}
    \caption{The mean $x$-drift variance for the tent map with $r=1.3$. The uniform initial distribution (blue) and a narrower initial distribution (red) centered at $x_0=0.5$ with width $w=0.01$ are compared. The former shows an initial decay of the drift variance, while the latter does not, since it does not go through a dissipative phase.}
    \label{fig_tentDriftVarComp}
\end{figure}

\begin{figure*}
    \centering
    \begin{subfigure}[t]{0.47\linewidth}
        \centering
        \includegraphics[width=\linewidth]{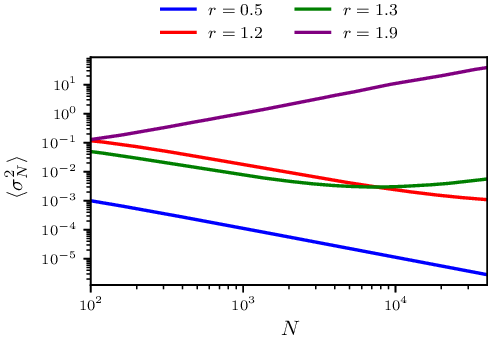}
        \caption{}
        \label{fig_MeanDriftVar}     
    \end{subfigure}
    \hspace{1em}
    \begin{subfigure}[t]{0.47\linewidth}
        \centering
        \includegraphics[width=\linewidth]{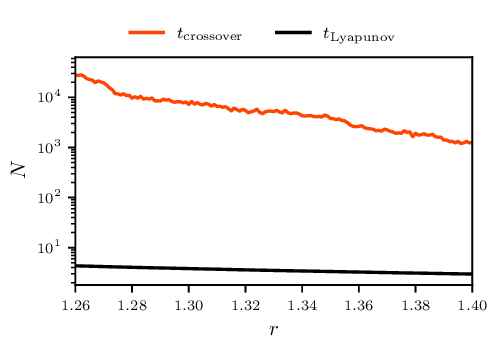}
        \caption{}
        \label{fig_mixingTime}
    \end{subfigure}
    \caption{(a) Mean $x$-drift variance with an initially uniform distribution as a function of $N$ for different values of $r$. (b) Comparison between crossover time for the uniform distribution and Lyapunov time for $1<r<\sqrt{2}$.}
\end{figure*}

The symmetric tent map \cite{crampin_1994} is defined on the interval $[0,1]$ and is given by the following piecewise linear function:
\begin{equation}
    f(x) = \begin{cases}
        r x, \ & x \leq \frac{1}{2},\\
        r (1 - x), \ & x > \frac{1}{2},
    \end{cases}
\end{equation}
where the coupling is restricted to $0 \leq r \leq 2$. This map has a fixed point attractor at $x^* = 0$ when $r < 1$, and is chaotic otherwise, with a Lyapunov time given by $t_{\text{Lyapunov}} = 1/\ln r$. In this chaotic regime, the phase space is characterized by a strange attractor with a band structure, which becomes narrower as $r$ approaches $1$ from above. Thus, we expect wide distributions to go through initial dissipative phases as they converge to the strange attractor. Such transient behavior is not captured by the Lyapunov time, since it does not depend on the initial distribution.

We now demonstrate how different initial distributions have different convergence properties, leading to different behaviors of the mean $x$-drift variance. In particular, we consider two different initial distributions, one wide and one narrow. Our first distribution is the uniform distribution over $[0,1]$, defined as
\begin{equation}
    \rho_{\mathrm{uniform}}(x) = \begin{cases}
        1, & 0 \leq x \leq 1,\\
        0, & \text{otherwise}.
        \end{cases}   
    \label{eq_uniform}
\end{equation}
For the second, narrower distribution, we consider one that is centered around a point $x_0$ with a width $w$, given by
\begin{equation}
    \rho(x; x_0,w) = \begin{cases}
        \frac{1}{w}, & |x-x_0| < w/2,\\
        0, & \text{otherwise}.
    \end{cases}
\end{equation}

In Fig. \ref{fig_tentDensity}, we compare the evolution of the time-averaged distribution $\bar{\rho}_N(x)$ (given by Eq. \ref{eq_avgMeasure}) for both initial distributions at $r=1.3$.
Note that at $r=1.3$, the strange attractor is restricted to two narrow bands in the phase space. The wider uniform distribution goes through an initial dissipative phase as it converges to the natural density. On the other hand, if we choose the parameters of the narrower distribution such that it lies completely inside the strange attractor, then its convergence is much faster, and there is no dissipative phase.

This transient dynamics is reflected in the behavior of the mean $x$-drift variance, plotted in Fig. \ref{fig_tentDriftVarComp} for both initial distributions. The variance of the uniform distribution shows an initial decay, but later growth. On the other hand, the narrower distribution has no such initial decay, and shows monotonic growth. At large times, both mean drift variances converge to the same asymptotic behavior, since both systems converge to the same natural distribution.

As discussed previously, the transition from decay to growth of the drift variance can be used to define a crossover time, which serves as a diagnostic for equilibration. This crossover is only observed when the strange attractor occupies a narrow portion of the phase space, which forces a wide initial distribution to go through a dissipative phase. See Fig. \ref{fig_MeanDriftVar}, which plots the mean $x$-drift variance for the uniform distribution at different values of $r$. As expected, when $r < 1$, a pure dissipative phase is observed, with the decay of the variance scaling as $1/N$. When $r=1.9$, the strange attractor occupies a significant portion of the phase space, and the system is immediately strongly chaotic. Hence, the variance exhibits linear growth. However, a significant dissipative phase and an eventual transition to chaos is observed when $1 < r < \sqrt{2}$. This crossover time is orders of magnitude larger that the Lyapunov time, as shown in Fig. \ref{fig_mixingTime}. 
Note that for $r=1.2$, which is the non-linearity studied in Fig. \ref{fig_tent}, the crossover time is beyond our computation time.

\subsection{Logistic Map}

\begin{figure*}
    \centering
    \begin{subfigure}[t]{0.45\textwidth}
        \centering
        \includegraphics[width=\linewidth]{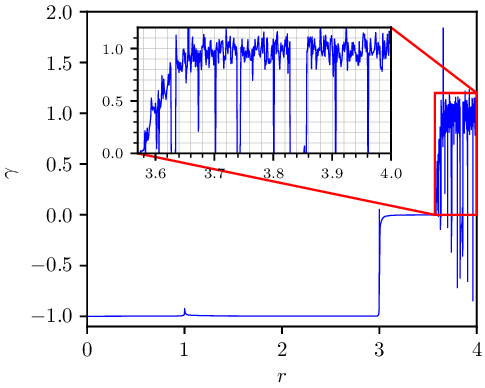}
        \caption{}
        \label{fig_diffExp}
    \end{subfigure}
    \hspace{1em}
    \begin{subfigure}[t]{0.45\textwidth}
        \centering
        \includegraphics[width=\linewidth]{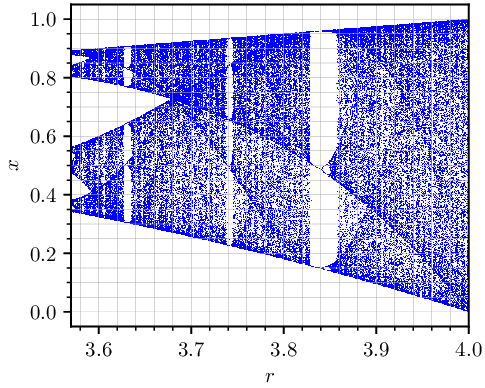}
        \caption{}
        \label{fig_bifurcation}
    \end{subfigure}
    \caption{(a) Diffusion exponent of the mean $x^2$-drift variance as a function of the non-linearity in the logistic map. (b) Bifurcation diagram for the chaotic region of the logistic map. The regions of stability can be identified by the gaps in the figure, which are all captured by the diffusion exponent.}
\end{figure*}

Next, we investigate the logistic map \cite{may_1976} in this section. The logistic map is a one-dimensional map defined on the interval $[0,1]$, and given by
\begin{equation}
    f(x) = rx(1-x),
\end{equation}
where the non-linearity is restricted to $0\leq r\leq 4$. For this example, we will consider our observable to be $O(x) = x^2$. We calculated and plot the asymptotic value of the diffusion exponent of the $x^2$-drift as a function of the non-linearity in Fig. \ref{fig_diffExp}. The mean $x^2$-drift variance is computed at each $r$ by averaging $\sigma_N^2$ over the uniform distribution (Eq. \ref{eq_uniform}), which is then fitted to a power law to extract $\gamma(\infty)$.  As expected, when $r<3$, the exponent is $-1$ due to the presence of a fixed point attractor. The logistic map contains limit cycle attractors when $3 < r < r_\infty \approx 3.56994\dots$, which leads to $\gamma$ being zero. Beyond $r_\infty$, the map exhibits chaos and the exponent $\gamma(\infty)$ is mostly one in this region. There are brief intervals of stability even within the chaotic region, such as the period-3 cycle that appears at $r = 1+\sqrt{8}$. These intervals are characterized by sharp drops in $\gamma(\infty)$ to and below zero. Some errors are expected in the computation of $\gamma(\infty)$ due to transient effects, but these diminish with increasing computation time. Comparing with the bifurcation diagram in Fig. \ref{fig_bifurcation}, we see that all regions of stability are accurately captured by the diffusion exponent. Similar results are obtained with the $0-1$ test of chaos, as shown in \cite{gottwald_2009}.

\subsection{Standard Map}
\label{sec_stdMap}

\begin{figure}
    \centering
    \includegraphics[width=0.95\linewidth]{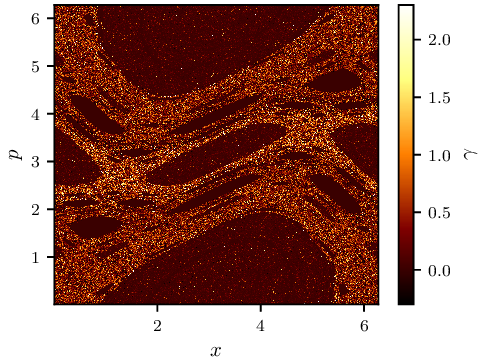}
    \caption{Phase-space of the Chirikov standard map at non-linearity $K=1$. The diffusion exponent $\gamma$ of the mean $xp$-drift variance is plotted as a function of $(x_0,p_0)$.}
    \label{fig_std}
\end{figure}


Finally, we study the Chirikov standard map \cite{chirikov_1971,chirikov_1979}, which is a two-dimensional map described by
\begin{align}
    &p_{n+1} = p_n + K\sin x_n \ &(\text{mod } 2\pi), \notag\\& x_{n+1} = x_n + p_{n+1} \ &(\text{mod } 2\pi).
\end{align}
This map displays a mixed phase space, featuring stable islands embedded within chaotic regions. For this map, we will consider the observable $O(x,p) = xp$, along with an initial distribution given by $\rho_0(x,p) = \delta(x-x_0)\delta(p-p_0)$. The evolution of the $xp$-drift variance on a particular ``sticky" trajectory is demonstrated in the supplementary video. Initially, the trajectory appears nearly periodic, indicating non-chaotic behavior. As it transitions into the chaotic region of the phase-space, the drift variance exhibits a sudden growth, signaling the onset of chaos.

In addition, we compute the asymptotic values of the diffusion exponent for each initial point $(x_0,p_0)$ and plot it as a function of the phase-space coordinates in Fig. \ref{fig_std}. Since the initial distribution is a delta function, we smooth out the drift variance by performing a time-average, given by
\begin{equation}
    \bar{\sigma}^2_N(x_0,p_0) = \frac{1}{N+1}\sum_{n=0}^N \sigma^2_n(x_0,p_0),
\end{equation}
and extract $\gamma$ by performing a linear fit on $\frac{\ln \bar{\sigma}^2_N(x_0,p_0)}{\ln N}$. Note that the mixed phase-space of the standard map is accurately reproduced, where regular regions are characterized by $\gamma=0$, while chaotic regions are characterized by $\gamma > 0$. Again, transient effects are expected to introduce some errors, but they diminish with increasing computation time.

\section{Conclusions}
\label{sec_conclusions}
In this paper, we have shown how the AGP probe for chaos can be extended and generalized to all classical dynamical systems. We first argue that weakly mixing systems are chaotic, since any errors in the initial conditions must eventually span the entire accessible phase-space of such systems. Thus, mixing captures long-time SDIC. This allows us to argue that chaotic systems can be identified by absolutely decaying correlations: that is
\begin{equation}
    \lim_{N\to\infty} \frac{1}{N+1}\sum_{n=0}^N |C_n(O_1,O_2)| = 0.
\end{equation}
By constructing the drift of an observable as a random walk, with each step given by the deviation of the observable from its mean, we show that decaying correlations result in indefinite growth of the drift variance, whereas quasi-periodic correlations result in its saturation. The strength of chaos is identified with the rate of decay of correlations: strong chaos corresponds to a fast decay, and therefore linear growth of the variance, while weak chaos corresponds to a slow decay and anomalous behavior. Dissipative dynamics is associated with a contraction of the phase-space density and is shown to correspond to a decreasing drift variance. We have also demonstrated the efficacy of this framework by providing numerical analysis of three discrete-time systems: namely, the tent map, logistic map, and the Chirikov standard map.

We list some open problems and unanswered questions that could be the subject of future work. While we have numerically studied some well-known one- and two-dimensional discrete-time systems, it remains to be seen how well this formalism works in more complex systems. Since the 0-1 test of chaos has been applied to much more complicated systems, a direct comparison with this test in a wide variety of systems would prove useful. Moreover, this method can be directly applied to quantum systems, making it a promising tool for studying both open and driven quantum dynamics. Additionally, while we have studied the second moment of the drift in this work, we expect higher moments to reveal finer details of the system. Specifically, it would be interesting to see if higher moments allow us to differentiate between different levels of the ergodic hierarchy.

Lastly, we hope that this work contributes to bridging the gap between quantum and classical chaos. As discussed previously, while classical chaos is defined as a sensitive dependence on the initial conditions, this notion cannot be directly translated to quantum systems. Instead, quantum chaos has been associated with the sensitivity to adiabatic deformations. In this paper, we show that these two definitions of chaos are equivalent. While the drift variance (or AGP norm) is known to probe sensitivity to adiabatic changes in both quantum \cite{pandey_2020} and classical systems \cite{lim_2024,karve_2025}, we show here that it also captures SDIC in classical dynamics. Since both forms of sensitivity ultimately stem from the decay of correlations, a unified definition of chaos in both quantum and classical systems may simply rest on the decay of two-point correlation functions.

\section{Acknowledgments}

The authors thank Anatoli Polkovnikov, Bernardo Barrera, and Guilherme Delfino for insightful discussions. The authors acknowledge the use of Boston University's Shared Computing Cluster (SCC) for numerical simulations.

\section{Supplementary Material}

This section describes the supplementary video. As discussed in Sec. \ref{sec_stdMap}, the video demonstrates the evolution of the $xp$-drift variance on a sticky trajectory. The system is initialized at the point $x_0 = 5.490807$ and $p_0 = 0$, with the non-linearity being $K = 1.05$. The trajectory is considered to be sticky since it initially shows near-periodic behavior, but eventually becomes chaotic. This transition from regular to chaotic dynamics is identified with a sudden growth in the drift variance.

\section{Conflict of Interest}

The authors have no conflicts to disclose.

\section{Data Availability}

The data that support the findings of this article are openly available \cite{karve_git}.

\appendix

\section{Ergodicity and Mixing}
\label{app_ergodicTheory}

In this section, we provide a brief overview of concepts from ergodic theory. More details can be found in standard textbooks \cite{walters_1982,sinai_1982,viana_2016}. We also demonstrate how mixing leads to late-time SDIC, and thus argue that such systems must be considered chaotic.

A discrete-time dynamical system is defined on a space $M$, with the map $f:M\to M$ defining time evolution $x_{n+1} = f(x_n)$ of states in $M$. The trajectory of any state $x\in M$ is given by the set $\{f^n(x) \mid n \geq 0\}$, where $f^n$ is the $n$th iterate of the map. We will also assume that $M$ is a metric space, and that the function $d(x,y)\geq 0$ defines a distance between any two states $x$ and $y$.

An ensemble of states of the system, described by a probability density $\rho_0(x)$, induces a probability measure $\rho_0$. More specifically, we will assume that the measure $\rho_0$ is defined on the measurable space $(M,\mathcal{B})$, where $\mathcal{B}$ is the Borel $\sigma$-algebra of $M$. The evolution of this measure with time can be described by
\begin{equation}
    \rho_n = \rho_0 \circ f^{-n}.
\end{equation}
Consequently, the natural measure, which we denote by $\mu$, is given by
\begin{equation}
    \mu = \lim_{N\to\infty} \frac{1}{N+1}\sum_{n=0}^N \rho_0 \circ f^{-n}.
\end{equation}

The property of ergodicity is defined with respect to a specific natural measure. A system is ergodic if a generic trajectory visits almost every point in the phase space with respect to the natural measure. In other words, the state space of an ergodic system cannot have independent subsets that do not visit each other \cite{oxtoby_1952}. Thus, only trivial subsets of $M$ can be invariant under the action of $f$, that is
\begin{equation}
    \forall A \in \mathcal{B} \text{ s.t. } f^{-1}(A) = A, \text{ either } \mu(A) = 0 \text{ or } 1.
\end{equation}
A crucial property of ergodic systems is that time averages can be replaced by phase space averages. For any $\mu$-integrable observable $O:M \to \mathbb{R}$, Birkhoff's ergodic theorem \cite{birkhoff_1931} states that for $\mu$-almost any $x \in M$,
\begin{equation}
    \lim_{N\to\infty}\frac{1}{N+1}\sum_{n=0}^{N} O(f^n(x)) = \int_M d\mu \ O.
\end{equation}

Although ergodicity requires a trajectory to visit almost the entire phase space, it does not imply SDIC. Trajectories initialized close to each other can stay close indefinitely. As an example, consider irrational rotations on the circle, given by the map
\begin{equation}
    \theta_{n+1} = \theta_n + 2\pi\alpha \ (\text{mod } 2\pi),
\end{equation}
where $\alpha$ is an irrational number. Each iteration of the above map simply rotates the state by an angle $2\pi\alpha$. This map is ergodic, since $\alpha$ being irrational guarantees that any trajectory will come indefinitely close to all points on the circle. However, a localized ensemble of initial conditions will always remain localized and will never spread out over the circle. Thus, any two trajectories in such an ensemble will remain at a constant separation.

Mixing systems, on the other hand, allow nearby trajectories to spread out. A system is called strongly mixing if for any two sets $A_1, \ A_2 \in \mathcal{B}$, any trajectory that starts in $A_1$ forgets its initial condition after many iterations, and the probability of finding this trajectory in $A_2$ only depends on the size of $A_2$: that is,
\begin{equation}
    \lim_{n\to\infty} \mu\left(A_1 \cap f^{-n}(A_2)\right) = \mu(A_1)\mu(A_2).
\end{equation}
This condition is slightly relaxed in weakly mixing systems by allowing this limit to hold on average. A system is called ``weakly mixing" if for any two sets $A_1, \ A_2 \in \mathcal{B}$,
\begin{equation}
    \lim_{N\to\infty} \frac{1}{N+1}\sum_{n=0}^N \left|\mu\left(A_1 \cap f^{-n}(A_2)\right) - \mu(A_1)\mu(A_2)\right| = 0.
\end{equation}
Thus, one expects nearby trajectories in mixing systems to diverge away from each other, since any initial ensemble must spread over the entire phase space.

Let us now define SDIC, which is a crucial feature of chaotic systems. As prescribed by Devaney's definition of chaos, SDIC corresponds to the fact that any blob of initial conditions, no matter how small, will eventually grow larger than a certain threshold value. We refer to this definition as ``Devaney SDIC."
\begin{definition}[Devaney SDIC]
    A system described by a bounded measure space $(M,\mathcal{B},\mu,f)$ with a metric $d$ is considered to depend sensitively on its initial conditions if $\exists \epsilon > 0$, such that for any open set $B\in\mathcal{B}$ with $\mu(B) > 0$, $\exists n \geq 0$ for which
    \begin{equation}
        \sup \left\{ d\left(f^n(x),f^n(y)\right) \mid x, y \in B\right\} \geq \epsilon.
    \end{equation}
\end{definition}
Note that Devaney's definition of chaos does not prescribe the value this threshold $\epsilon$ must take. If $\epsilon$ is much smaller than the dimensions of the space $M$, then the degree of unpredictability is minimal. Therefore, we will define a slightly stronger version of SDIC, which requires not only that nearby trajectories diverge, but that they separate as much as possible within the space. We will refer to this condition as ``global SDIC." Given a blob of initial conditions $B$, we will quantify the separation between trajectories after $n$ iterations by the diameter of $f^n(B)$, where the diameter of any open set is defined as
\begin{equation}
    \text{dia}(B) = \sup\{d(x,y) \mid x,y \in B\}.
\end{equation}
We will primarily be concerned with initial conditions in the support of $\mu$, where
\begin{equation}
    \text{supp}(\mu) = \{x \in M \mid \forall N_x \in \mathcal{B}, \ \mu(N_x) > 0\}.
\end{equation}
Then, our SDIC condition will require any set of initial conditions to eventually span the entire support.

\begin{definition}[Global SDIC]
    Consider a system described by a bounded measure space $(M,\mathcal{B},\mu,f)$ with a metric $d$. This system is considered to  depend sensitively on its initial conditions if any open set $B \subseteq \mathrm{supp}(\mu)$ eventually spans the entire supported space: that is,
    \begin{equation}
        \sup\{\mathrm{dia}(f^n(B)) \mid n \geq 0\} = \mathrm{dia}(\mathrm{supp}(\mu)).
    \end{equation}
\end{definition}

Clearly, SDIC is not a sufficient condition for chaos. As noted in \cite{werndl_2009}, one can consider an expanding map, such as $x_{n+1} = c x_n$, with $c > 1$ and $x \in (0,\infty)$, where nearby trajectories diverge exponentially fast, and yet the system remains predictable. Notably, such a map is not ergodic, since a generic trajectory is not dense in the phase-space. Thus, for a system to be chaotic, we will also require it to be ergodic.

Mixing systems are by definition ergodic, and are also known to exhibit Devaney SDIC, as shown in \cite{werndl_2009, he_2004}. We will now show that weak mixing also implies global SDIC.


\begin{figure}
    \centering
    \includegraphics[width=0.9\linewidth]{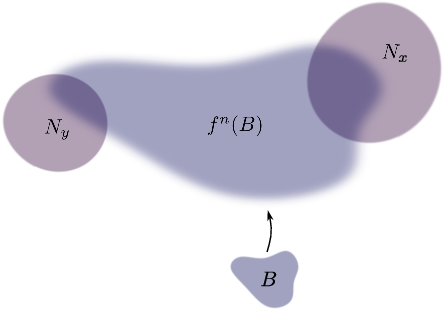}
    \caption{Weak mixing implies global SDIC -- any open set of states $B$ will eventually spread over the support of $\mu$. Thus, no matter how small the uncertainty in the initial condition, it always becomes large in the future.}
    \label{fig_mixingSDIC}
\end{figure}

\begin{theorem}
    Consider a system described by a bounded measure space $(M,\mathcal{B},\mu,f)$, with a metric $d$. This system exhibits global SDIC if it is weakly mixing.
\end{theorem}

\begin{proof}
    Consider an open set $B \subseteq \mathrm{supp}(\mu)$ and any two points $x,y \in \mathrm{supp}(\mu)$. We will show that for any two open neighborhoods $N_x$ and $N_y$, there exists an $n\geq 0$ for which $f^n(B)$ simultaneously intersects both neighborhoods (see Fig. \ref{fig_mixingSDIC} for an illustration). A standard result in ergodic theory states that the product space of a weakly mixing system is also weakly mixing, and therefore, ergodic \cite{einsiedler_2012}. Thus, there exists an $n\geq 0$ such that $f^n\times f^n (B\times B)$ visits $B_1\times B_2$. This implies our claim that $f^n(B)$ simultaneously intersects $N_x$ and $N_y$. Since this is true for all neighborhoods of $x$ and $y$, $\sup\{\mathrm{dia}(f^n(B)) \mid n \geq 0\}$ must be at least as big as $d(x,y)$. And since this is true for all $x$ and $y$ in the support, the global SDIC condition must also be true.
\end{proof}

Note that this SDIC condition does not require nearby trajectories to separate exponentially fast, and therefore, is a weaker condition than requiring a positive Lyapunov exponent. Nevertheless, even in systems where nearby trajectories diverge sub-exponentially, information about the initial condition is lost at large times. Therefore, we will consider global SDIC along with ergodicity, and consequently weak mixing, to be a characteristic feature of chaotic systems.

\section{Non-Chaotic Attractors}
\label{app_attractors}

We will first study the behavior of the observable drift in systems with non-chaotic attractors -- specifically, fixed point and limit cycle attractors. These are simple examples of non-ergodic systems, and it is easy to see how a fixed point attractor corresponds to dissipative dynamics, while a limit-cycle attractor asymptotically approaches regular dynamics. The analysis will be focused on one-dimensional maps, but should be generalizable to higher dimensional systems.

\subsection{Fixed Point Attractor}
\label{app_FixedPtAttractor}

We prove the following proposition for trajectories in one-dimensional maps that converge to a fixed point attractor. 
\begin{theorem}
    Consider a one-dimensional map $f$ with a stable fixed point $x^*$, such that $|f'(x^*)| < 1$. If a trajectory $\{f^n(x)\}$ converges to this fixed point, then the $O$-drift variance for some observable O, given by Eq. \ref{eq_driftVar}, converges to zero as $N \to \infty$. Moreover, the variance scales as $\sigma_N^2(x) \sim \mathcal{O}(1/N)$ for large $N$.
\end{theorem}

Before we prove this proposition, we state some useful lemmas.

\begin{lemma}
    \label{lemma_expConv}
    If a trajectory $\{f^n(x)\}$ converges to a fixed point $x^*$, then it converges exponentially fast. That is, $|O(f^n(x)) - O(x^*)| < C\lambda^n$, for some constants $C > 0$ and $0 < \lambda < 1$.
\end{lemma}

\begin{proof}
    Since $x^*$ is a stable fixed point, $|f'(x^*)| < 1$. Therefore, assuming that $f'$ is continuous, we can find a compact neighborhood $U$ around $x^*$ such that $|f'(y)| < \lambda < 1$ for all $y \in U$. Since the trajectory $f^n(x)$ converges to $x^*$, we can find a $M$ such that $\forall n \geq M$, $f^n(x) \in U$. For all such $n$, the mean value theorem allows us to write:
    \begin{align}
        |f^n(x) - x^*| &= |f(f^{n-1}(x)) - f(x^*)| \notag\\&= |f'(y_n)| |f^{n-1}(x) - x^*|,
    \end{align}
    where $y_n \in U$. Consequently, $|f'(y_n)| < \lambda$, and
    \begin{equation}
        |f^n(x) - x^*| < \lambda |f^{n-1}(x) - x^*| < \lambda^{n-M}|f^M(x)-x^*|.
    \end{equation}
    Choosing $D = \sup \bigl\{\frac{|f^n(x)-x^*|}{\lambda^n} \bigm| 0 \leq n \leq M\bigr\}$, we get
    \begin{equation}
        |f^{n}(x) - x^*| < D\lambda^{n}.
    \end{equation}
    Since $O$ is a smooth function of $x$, its derivative must be bounded on $U$. That is,$\forall y \in U$, $|O'(y)| \leq \frac{C}{D}$, for some positive constant $C$. Then, the mean value theorem gives us:
    \begin{equation}
        |O(f^n(x))-O(x^*)| = |O'(z_n)||f^n(x)-x^*|,
    \end{equation}
    for some $z_n \in U$. Therefore,
    \begin{equation}
        |O(f^n(x))-O(x^*)| \leq \frac{C}{D}|f^n(x)-x^*| \leq C \lambda^n.
    \end{equation}
\end{proof}

\begin{lemma}
    \label{lemma_driftExpConv}
    The drift on such a trajectory converges exponentially fast to a finite value $\Delta_\infty(x)$ in the limit $N\to\infty$.
\end{lemma}
\begin{proof}
    We note that $\Delta_N(x)$ converges absolutely via the comparison test, since
    \begin{align}
        \sum_{n=0}^N |O(f^n(x))-O(x^*)| \leq \sum_{n=0}^N C\lambda^n = C\left(\frac{1-\lambda^{N+1}}{1-\lambda}\right).
    \end{align}
    And therefore,
    \begin{equation}
        \lim_{N\to\infty} |\Delta_N(x)| \leq \frac{C}{1-\lambda}.
    \end{equation}
    Further,
    \begin{align}
        |\Delta_N(x) - \Delta_\infty(x)| &\leq \sum_{n=N+1}^\infty |O(f^n(x))-O(x^*)| \notag\\&\leq \frac{C\lambda^{N+1}}{1-\lambda},
    \end{align}
    which proves exponential convergence.
\end{proof}

Finally, we use the above two lemmas to prove our original assertion.
\begin{proof}
    Let us define $\tilde{\Delta}_n(x) = \Delta_n(x) - \Delta_\infty(x)$. Then, the drift variance can be written as:
    \begin{equation}
        \sigma_N^2(x) = \frac{1}{N+1}\sum_{n=0}^N \tilde{\Delta}_n^2(x) - \left(\frac{1}{N+1}\sum_{n=0}^N \tilde{\Delta}_n(x)\right)^2.
    \end{equation}

    Both terms in the above equation converge to zero as $N\to\infty$, since:
    \begin{subequations}
        \begin{equation}
            \left|\frac{1}{N+1}\sum_{n=0}^N \tilde{\Delta}_n(x)\right| \leq \frac{C\lambda}{(1-\lambda)^2(N+1)},
        \end{equation}
        \begin{equation}
            \left|\frac{1}{N+1}\sum_{n=0}^N \tilde{\Delta}_n^2(x)\right| \leq \frac{C\lambda^2}{(1-\lambda)^2(1-\lambda^2)(N+1)}.
        \end{equation}
    \end{subequations}
    Therefore, we conclude that $\displaystyle\lim_{N\to\infty}\sigma^2_N(x) = 0$. Further note that while $\frac{1}{N+1}\sum_{n=0}^N \tilde{\Delta}_n(x)$ can decay faster than $\mathcal{O}(1/N)$, the term $\frac{1}{N+1}\sum_{n=0}^N \tilde{\Delta}_n^2(x)$ has to decay as $\mathcal{O}(1/N)$, since it is bounded from below by a term that scales as $1/N$:
    \begin{equation}
        \frac{\tilde{\Delta}_0^2(x)}{N+1} \leq \left|\frac{1}{N+1}\sum_{n=0}^N \tilde{\Delta}_n^2(x)\right|,
    \end{equation}
    for every $N \geq 0$. Therefore, $\sigma_N^2(x)$ shows $1/N$ scaling at large $N$.
\end{proof}

\subsection{Limit-cycle Attractor}
\label{app_limitCycle}

Any trajectory that converges to a limit-cycle attractor will eventually cycle through a finite number of points. We demonstrate that in this case, the drift variance converges to a finite value, and does so at a rate $\sim\mathcal{O}(1/N)$. To demonstrate this, we will consider the $O$-drift variance on a period-2 limit-cycle, but the claim can be easily generalized to any periodic orbit.

\begin{theorem}
    Consider a one-dimensional map $f$ with a stable period-2 limit-cycle $\{x_0^*, x_1^*\}$. The $O$-drift variance on any trajectory $\{f^n(x)\}$ that converges to this limit-cycle converges to a positive value $\delta^2/4$, where $\delta = \left|\frac{O(x_1^*)-O(x_0^*)}{2}\right|$. Further, the rate of convergence is $|\sigma_N^2(x) - \delta^2/4| \sim \mathcal{O}(1/N)$.
\end{theorem}

\begin{proof}
    Without loss of generality, let us assume that $O(x_1^*) > O(x_0^*)$, and that even iterates of the trajectory converge to $x_0^*$, while odd iterates converge to $x_1^*$. For simplicity, we will denote $O(f^n(x))$ by $O_n$ henceforth. Since $\{x_0^*,x_1^*\}$ form a stable period-2 cycle, both $x_0^*$ and $x_1^*$ are stable fixed points of the second iterate of the map $g(x) = f(f(x))$, with $|g'(x_0^*)| < 1$ and $|g'(x_1^*)| < 1$. Following Lemma \ref{lemma_expConv}, we have:
    \begin{equation}
        |O_{2m} - O_0^*| < C\lambda^{2m}, \quad |O_{2m+1} - O_1^*| < C\lambda^{2m+1},
    \end{equation}
    for some $C > 0$ and $0<\lambda<1$. That is, the even and odd iterates of $f$ converge exponentially fast to $O_0^*$ and $O_1^*$ respectively. Further, it is easy to see that the long-time average of $O$ is simply:
    \begin{equation}
        \bar{O} = \lim_{N\to\infty}\frac{1}{N+1}\sum_{n=0}^N O_n = \frac{O_0^*+O_1^*}{2}.
    \end{equation}

    Then,
    \begin{subequations}
        \begin{equation}
            O_{2m}-\bar{O} = O_{2m} - O_0^* - \delta,
        \end{equation}
        \begin{equation}
            O_{2m+1}-\bar{O} = O_{2m+1} - O_1^* + \delta,
        \end{equation}
    \end{subequations}
    where $\delta = \frac{O_1^* - O_0^*}{2}$. Consequently, the drift can be expressed as:
    \begin{equation}
        \Delta_N(x) = \zeta_N(x) - \frac{1+(-1)^N}{2}\delta,
    \end{equation}
    where,
    \begin{equation}
        \zeta_N(x) = \sum_{m=0}^{\lfloor \frac{N}{2} \rfloor} (O_{2m}-O_0^*) + \sum_{m=0}^{\lfloor \frac{N-1}{2} \rfloor} (O_{2m+1}-O_1^*).
    \end{equation}
    Thus, the drift variance can be expressed as
    \begin{equation}
        \sigma_N^2(x) = \frac{\delta^2}{4} + \xi_N^2(x) - \frac{\delta}{N+1}\sum_{n=0}^N \zeta_n(x) (-1)^n,
        \label{eq_driftVarLimitCycle}
    \end{equation}
    where $\xi_N(x)$ is the variance of $\zeta_N(x)$,
    \begin{equation}
        \xi_N(x) = \frac{1}{N+1}\sum_{n=0}^N \zeta_n^2(x) - \left(\frac{1}{N+1}\sum_{n=0}^N \zeta_n(x)\right)^2.
    \end{equation}
    Following Lemma \ref{lemma_driftExpConv}, $\zeta_N(x)$ must converge to a finite value at an exponentially fast rate,
    \begin{equation}
        |\zeta_N(x)-\zeta_\infty(x)| \leq \frac{C\lambda^{N+1}}{1-\lambda},
    \end{equation}
    and $\xi_N(x)$ must converge to zero as $\mathcal{O}(1/N)$. Finally, the last term in Eq. \ref{eq_driftVarLimitCycle} can also be shown to converge to zero, since it is equivalent to $\displaystyle\lim_{N\to\infty}\frac{\delta}{N+1}\sum_{n=0}^N(-1)^{n}(\zeta_n(x)-\zeta_\infty(x))$, and
    \begin{equation}
        \left|\frac{\delta}{N+1}\sum_{n=0}^N(-1)^{n}(\zeta_n(x)-\zeta_\infty(x))\right| \leq \frac{\delta C\lambda}{(1-\lambda)^2(N+1)}.
    \end{equation}
    Thus, this term converges to zero at a rate faster than $\mathcal{O}(1/N)$. Therefore, the drift variance converges to $\delta^2/4$, and does so as $\sim \mathcal{O}(1/N)$.
\end{proof}

\section{Ergodic Systems}
\label{app_ergodic}

Ergodic systems are characterized by invariant densities where phase space averages coincide with time averages on typical trajectories. As such, we will focus on drift variances on these invariant densities. We claim that the following relation between the mean drift variance and the auto-correlation function of the corresponding observable holds:
\begin{widetext}
\begin{theorem}[Fluctuation-Dissipation Relation]
    Given a discrete map $f$ with an ergodic natural measure and a measurable observable $O$, the $O$-drift variance averaged over the natural distribution is related to the auto-correlation function, $C_n(O,O)$, as:
    \begin{align}
        &\langle\sigma^2_N\rangle = 
        \sum_{r = -N}^N \frac{(N-|r|)(N+1-|r|)(N+2-|r|)}{6(N+1)^2} C_r(O,O).
        \label{eq_flucDiss}
    \end{align}
\end{theorem}

\begin{proof}
    We first note that:
    \begin{equation}
        \langle \Delta_{N_1}(x)\Delta_{N_2}(x)\rangle = \sum_{n_1=0}^{N_1}\sum_{n_2=0}^{N_2} C_{n_1-n_2}(O,O).
    \end{equation}

    If we denote $n_1-n_2 = r$, then the number of $C_r(O,O)$ terms in the above equation is $\text{min}(N_1+1-|r|,N_2+1)$ if $r \geq 0$, and $\text{min}(N_1+1,N_2-|r|)$ if $r < 0$. Using the fact that
    \begin{align}
        &\sum_{N_1 = r}^N\sum_{N_2=0}^N \text{min}(N_1+1-r,N_2+1) =
        \frac{(N+1-|r|)(N+2-|r|)(2N+3+|r|)}{6},
    \end{align}

    we can write:
    \begin{align}
        &\left\langle \left(\frac{1}{N+1}\sum_{n=0}^N \Delta_n(x)\right)^2\right\rangle = 
        \sum_{r=-N}^N\frac{(N+1-|r|)(N+2-|r|)(2N+3+|r|)}{6(N+1)^2} 
        C_r(O,O).
    \end{align}

    A similar argument gives us:
    \begin{align}
        &\left\langle \frac{1}{N+1}\sum_{n=0}^N \Delta_n^2(x)\right\rangle = 
        \sum_{r=-N}^N\frac{(N+1-|r|)(N+2-|r|)}{2(N+1)} C_r(O,O).
    \end{align}

    The above two equations can be combined to give us the fluctuation-dissipation relation. A derivation of the continuous time version can be found in \cite{karve_2025}.
\end{proof}
\end{widetext}

\subsection{Purely Non-Mixing Systems}
\label{app_ergodicNotMixing}

With this fluctuation-dissipation relation, we can deduce the behavior of the mean drift variance given the form of the auto-correlation function. We first focus on non-weakly-mixing systems with pure point spectra. Correlations in such systems are quasi-periodic, and can take the form
\begin{equation}
    C_r(O,O) = \sum_{\omega} A_\omega \cos(\omega r),
\end{equation}
with a discrete set of frequencies $\omega$. Inserting this into Eq. \ref{eq_flucDiss}, we find that
\begin{equation}
    \langle \sigma^2_N\rangle = \sum_{\omega} A_\omega\left[\frac{1}{4\sin^2\frac{\omega}{2}} - \frac{1-\cos(N\omega)}{8N^2\sin^4\frac{\omega}{2}}\right] +\mathcal{O}\left(\frac{1}{N^3}\right).
\end{equation}
Thus, the drift variance approaches a finite value of $\sum_\omega \frac{A_\omega}{4\sin^2\frac{\omega}{2}}$ as $N\to\infty$. One can easily check that this is equivalent to the expression found for limit-cycles if $\omega$ corresponds to a single frequency. However, the convergence is faster here, since the second term in the above expression decays as $\mathcal{O}(1/N^2)$.

\subsection{Weakly Mixing Systems}
\label{app_weakMixing}

Weakly mixing systems are characterized by decaying averaged absolute correlations; that is,
\begin{equation}
    \lim_{r\to\infty} \frac{1}{r}\sum_{s=1}^r |C_s(O,O)| = 0.
\end{equation}
Thus, to study the behavior of the drift in such systems, we define the averaged correlation, $\bar{C}_r(O,O)$, as
\begin{equation}
    \bar{C}_r(O,O) = \frac{1}{r} \sum_{s=1}^r C_s(O,O),
\end{equation}
with $|r| \geq 1$. Clearly, this quantity must also decay to zero as $r\to\infty$ in weakly mixing systems, since
\begin{equation}
    |\bar{C}_r(O,O)| \leq \frac{1}{r}\sum_{s=1}^r |C_s(O,O)|.
\end{equation}
Assuming that the correlation is an even function, $C_{-r}(O,O) = C_r(O,O)$, it is easy to show that the fluctuation-dissipation relation can be expressed in terms of the averaged correlation as
\begin{align}
    \langle\sigma^2_N\rangle = & \ \frac{N(N+2)}{6(N+1)}C_0(O,O) \notag\\&+ \sum_{r=1}^N \frac{r(N-r)(N+1-r)}{(N+1)^2}\bar{C}_r(O,O).
    \label{eq_driftVarWeakMix}
\end{align}
The first term in the above expression always grows as $\mathcal{O}(N)$, but the behavior of the second term depends on the exact rate of decay of the correlation. Assuming that the averaged correlation is bounded by a polynomial decay rate, $|\bar{C}_r(O,O)| \leq \frac{K}{r^\alpha}$ for $r\geq 1$, we find that
\begin{align}
    &\left|\xi_N\right| \leq K\left\{\frac{N}{N+1}H_N^{(\alpha-1)} - \frac{2N+1}{(N+1)^2}H_N^{(\alpha-2)}\right.\notag\\&\qquad\qquad\qquad\qquad\qquad\qquad\qquad\left. + \frac{1}{(N+1)^2}H_N^{(\alpha-3)}\right\},
\end{align}
where $\displaystyle \xi_N = \sum_{r=1}^N \frac{r(N-r)(N+1-r)}{(N+1)^2}\bar{C}_r(O,O)$ is the second term in Eq. \ref{eq_driftVarWeakMix} and $\displaystyle H_{N}^{(\alpha)} = \sum_{r=1}^N \frac{1}{r^\alpha}$ is the generalized harmonic number. Using the fact that as $N\to\infty$,
\begin{equation}
    H_{N}^{(\alpha)} \sim \begin{cases}
        \text{constant}, & \alpha > 1\\
        \ln N, & \alpha = 1\\
        \displaystyle\frac{N^{1-\alpha}}{1-\alpha}, & \alpha < 1,
        \end{cases}
\end{equation}
we deduce the following. When $\alpha > 1$, $\xi_N$ grows sub-linearly with $N$, implying that the mean drift variance asymptotically approaches the linearly growing term $\frac{N(N+2)}{6(N+1)}C_0(O,O)$ -- a characteristic of strongly chaotic systems. On the other hand, if $\alpha < 1$, then $\xi_N$ can grow super-linearly, but not faster than $\mathcal{O}(N^{2-\alpha})$. In general, this implies the growth of the drift variance is super-linear and is bounded between $\mathcal{O}(N)$ and $\mathcal{O}(N^{2-\alpha})$, unless specific cancellations lead to sub-linear growth in very special cases. For example, specific cancellations leading to no growth do happen in purely non-mixing systems, where $\bar{C}_r(O,O)$ can show polynomial decay even if $\displaystyle\frac{1}{r}\sum_{s=1}^r |C_s(O,O)|$ does not decay. We classify systems that show anomalous growth of the mean drift variance as weakly chaotic.

A similar analysis can be performed for strongly mixing systems. Assuming that the correlation function itself is bounded by a polynomial decay, $|C_r(O,O)| \leq \frac{K}{r^\alpha}$, we again observe weak chaos if $\alpha \leq 1$, and strong chaos otherwise. Furthermore, strong chaos is also observed when the correlation function decays exponentially, $C_r(O,O) = K e^{-\lambda |r|}$, for example, in Anosov systems. The asymptotic behavior of the mean drift variance is
\begin{equation}
    \langle\sigma^2_N\rangle = \frac{K}{6} \coth\left(\frac{\lambda}{2}\right) N.
\end{equation}

\begin{figure*}
    \centering
    \begin{subfigure}[t]{0.45\linewidth}
        \centering
        \includegraphics[width=\linewidth]{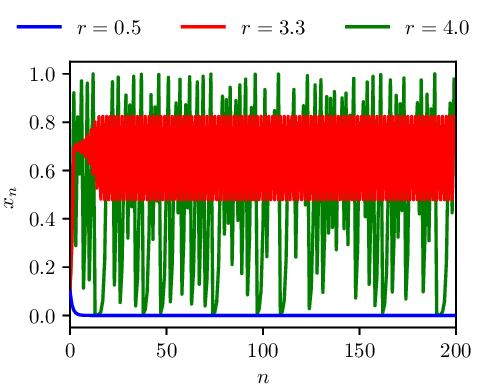}
        \caption{}
        \label{fig_logisticTraj}
    \end{subfigure}
    \begin{subfigure}[t]{0.45\linewidth}
        \centering
        \includegraphics[width=\linewidth]{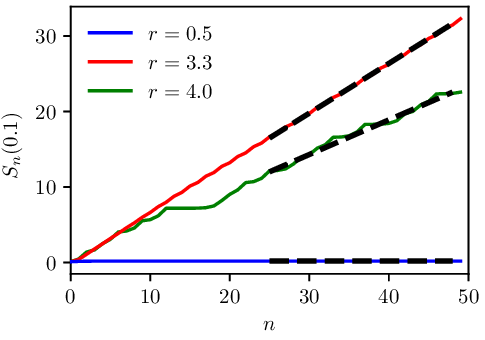}
        \caption{}
        \label{fig_logisticPartSum}
    \end{subfigure}
    \begin{subfigure}[t]{0.45\linewidth}
        \centering
        \includegraphics[width=\linewidth]{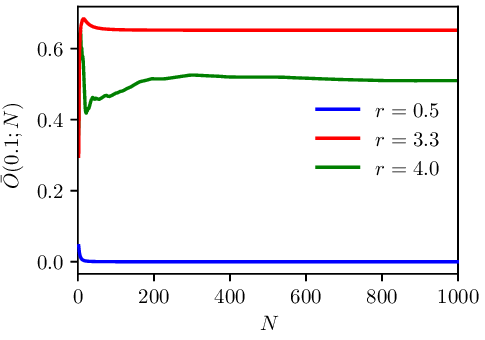}
        \caption{}
        \label{fig_logisticSlopes}
    \end{subfigure}
    \begin{subfigure}[t]{0.45\linewidth}
        \centering
        \includegraphics[width=\linewidth]{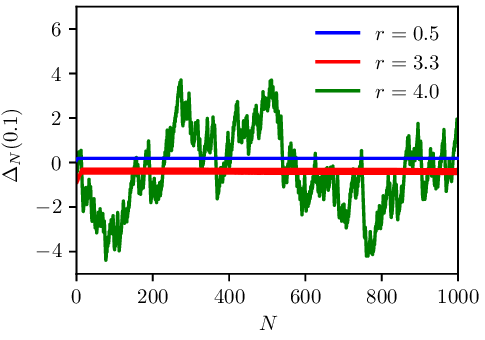}
        \caption{}
        \label{fig_logisticDrift}
    \end{subfigure}
    \begin{subfigure}[t]{0.45\linewidth}
        \centering
        \includegraphics[width=\linewidth]{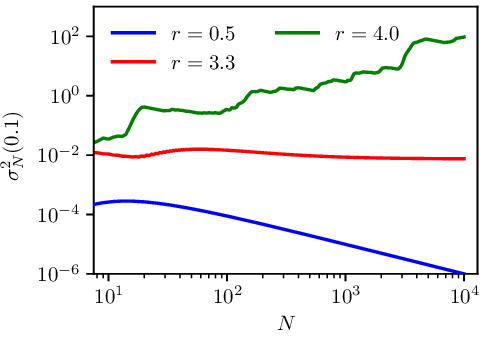}
        \caption{}
        \label{fig_logisticDriftVar}
    \end{subfigure}
    \begin{subfigure}[t]{0.45\linewidth}
        \centering
        \includegraphics[width=\linewidth]{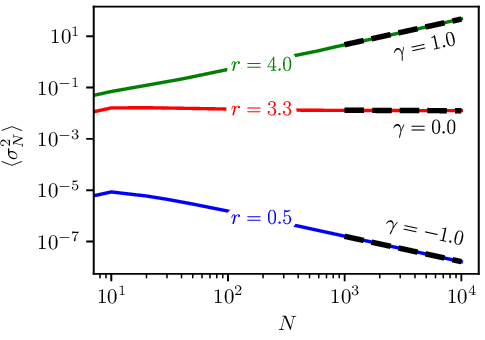}
        \caption{}
        \label{fig_logisticDriftVarAvg}
    \end{subfigure}
    \caption{(a) Three trajectories in the logistic map, all initialized at $x=0.1$, with $r = 0.5$, $3.3$, and $4$. The first and second trajectories converge to a fixed point and limit cycle attractor respectively, while the third one is chaotic. (b) The partial sum $S_n$ vs $n$ for these trajectories. The black dashed lines denote a linear fit. (c) The slopes obtained from this linear fit, $\bar{O}(x;N)$, as a function of $N$. They are observed to converge to their asymptotic values within the first few hundred iterations. (d) The corresponding $x^2$-drifts. (e) The corresponding $x^2$-drift variances. (f) The $x^2$-drift variances averaged over a uniform initial distribution. The dashed black lines correspond to a linear fit of $\ln\langle\sigma_N^2\rangle$ vs $N$, from which the diffusion exponent is extracted.}
\end{figure*}

\section{The 0-1 Test for Chaos}
\label{app_01Test}

The 0-1 test for chaos was formulated by Gottwald, \textit{et al.}\cite{gottwald_2004,gottwald_2005,gottwald_2009} as a binary tool for identifying chaos. We provide a brief review of the test in this section, and discuss how it relates to our framework.

Consider a discrete-time dynamical systems that produces a time-series $\phi(n)$, such that it is centered around zero. This time-series is used to drive an auxiliary two-dimensional system, via
\begin{subequations}
    \begin{equation}
        p(n+1) = p(n) + \phi(n)\cos cn,
    \end{equation}
    \begin{equation}
        q(n+1) = q(n) + \phi(n)\sin cn,
    \end{equation}
    \label{eq_01defn}
\end{subequations}
where $c \in (0,2\pi)$ is some constant. The authors propose that chaos in the original system can be probed by studying diffusion in this auxiliary system.

Diffusion in this auxiliary system is quantified through the mean-squared displacement (MSD), defined as
\begin{equation}
    M(n) = \lim_{N\to\infty} \frac{1}{N} \sum_{j=1}^N [p(j+n) - p(j)]^2 + [q(j+n) - q(j)]^2.
\end{equation}
By plugging in Eq. \ref{eq_01defn} into the above equation, it is easy to show that the MSD relates to the autocorrelation of $\phi$ as
\begin{equation}
    M(n) = \sum_{j=-(n-1)}^{(n-1)} (n - |j|) C_{\phi,\phi}(|j|) \cos cj, 
\end{equation}
where
\begin{equation}
    C_{\phi,\phi}(j) = \lim_{N \to \infty} \frac{1}{N} \sum_{n=1}^N \phi(n+j)\phi(n).
\end{equation}
Note the similarities with Eq. \ref{eq_flucDissText}.

One can now analyze how the MSD depends on the behavior of the correlation function $C_{\phi,\phi}$. If the correlation is quasi-periodic,
\begin{equation}
    C_{\phi,\phi}(j) = \sum_{\omega} A_\omega \cos\omega j,
\end{equation}
then,
\begin{align}
    M(n) = \sum_\omega A_\omega &\left[\frac{(1 - \cos c\cos\omega)(1 - \cos cn \cos \omega n)}{(\cos c - \cos\omega)^2}\right.  \notag\\&\qquad\qquad\qquad- \left.\frac{\sin c\sin \omega \sin cn \sin \omega n}{(\cos c - \cos\omega)^2}\right].
\end{align}
Thus, as long as the driving frequency $c$ does not resonate with the system, the MSD stays bounded indefinitely.

On the other hand, if the correlation decays exponentially fast, i.e. $C_{\phi,\phi}(j) = Ae^{-\lambda|j|}$, then the asymptotic behavior of the MSD is linear:
\begin{equation}
    M(n\to\infty) \sim An \frac{\sinh\lambda}{\cosh\lambda - \cos c}.
\end{equation}

A linearly growing MSD is again observed when the correlation exhibits a polynomial decay, $C_{\phi,\phi}(j) = A/|j|^\alpha$, for $j \neq 0$. The asymptotic behavior of the MSD is then given by
\begin{equation}
    M(n\to\infty) \sim C_{\phi,\phi}(0)n + An[\mathrm{Li}_\alpha(e^{ic}) + \mathrm{Li}_\alpha(e^{-ic})],
\end{equation}
where $\mathrm{Li}_n(z) = \sum_{k=1}^\infty z^k/k^n$ refers to the polylog function. Thus, a linearly growing MSD is indicative of a chaotic system. However, this test does not distinguish between weak and strong chaos.

\section{Numerical Details}
\label{app_numerics}

In this section, we provide details on the computation of the drift variance and the diffusion exponent. We will focus on three different trajectories in the logistic map, initialized at $x = 0.1$ with $r = 0.5$, $3.3$, and $4$. These trajectories are dissipative, regular, and chaotic, respectively, as shown in Fig. \ref{fig_logisticTraj}.

We consider the observable $O(x) = x^2$ for the logistic map. While the exact value of $\bar{O}(x) = \lim_{N\to\infty}\frac{1}{N+1} \sum_{n=0}^N O(f^n(x))$ cannot be computed numerically, there are various methods to estimate it. One method would be to directly compute the average of the observable over a long time. Another method, which we employ throughout this paper, is to first compute the partial sum
\begin{equation}
    S_n(x) = \sum_{m=0}^n O(f^m(x)),
\end{equation}
and extract $\bar{O}(x)$ as the slope of $S_n$ vs $n$, since
\begin{equation}
    S_n(x) = (n+1) \bar{O}(x) + \Delta_n(x).    
\end{equation}
See Fig. \ref{fig_logisticPartSum}, where we demonstrate this for the three trajectories under consideration. A simple linear regression \cite{draper_1998} yields
\begin{equation}
    \bar{O}(x) \approx \frac{12}{(N+1)(N+2)}\sum_{n=0}^N \left(\frac{n}{N} - \frac{1}{2}\right)S_n(x).
\end{equation}
Evidently, this estimate depends on the computation time $N$, and we make this dependence explicit by denoting the average as $\bar{O}(x; N)$. See Fig. \ref{fig_logisticSlopes}, which shows the dependence of the average on $N$. It is clear that for the logistic map, a good estimate of $\bar{O}$ can be obtained within the first few hundred iterations.

Next, the drift is computed using this estimation of the average, as
\begin{equation}
    \Delta_N(x) = \sum_{n=0}^N \left(O(f^n(x)) - \bar{O}(x;N)\right),
\end{equation}
and its variance is computed from Eq. \ref{eq_driftVar}. See Figs. \ref{fig_logisticDrift} and \ref{fig_logisticDriftVar}. The variance is then averaged over the initial distribution; we consider the uniform distribution in Fig. \ref{fig_logisticDriftVarAvg}. And finally, the diffusion exponent is computed by performing another linear regression on $\ln \langle \sigma_N^2\rangle$ vs $\ln N$.

\bibliographystyle{aipnum4-2}
\bibliography{references.bib}

\end{document}